%% file: Main.tex
\RequirePackage{fix-cm}
\documentclass[smallextended]{svjour3}       
\smartqed  
\usepackage{graphicx}
\usepackage{subfigure}
\usepackage[font=small]{caption}

\usepackage{color}
\usepackage{soul}

\usepackage{array} 
\usepackage{verbatim}
\usepackage{hyperref}
\usepackage{fancyhdr}

\usepackage[round]{natbib}

\usepackage[export]{adjustbox}

\usepackage[leftcaption]{sidecap}

\usepackage{amsmath}
\usepackage{mathtools}
\usepackage{amsfonts}
\usepackage{amssymb}
\usepackage{algorithm}
\usepackage[noend]{algpseudocode} 
\makeatletter
\def\BState{\State\hskip-\ALG@thistlm}
\makeatother

\journalname{J Comb Optim}
\begin{document}

\title{Team Selection For Prediction Tasks
}


\author{MohammadAmin Fazli         \and
        Azin Ghazimatin \and
       Jafar Habibi \and
       Hamid Haghshenas
}


\institute{MohammadAmin Fazli \at
              Department of Computer Engineering, Sharif University of Technology, Tehran, Iran \\
              \email{fazli@ce.sharif.edu}           
           \and
           Azin Ghazimatin \at
               \email{azinghazimatin@ce.sharif.edu}
	\and
	Jafar Habibi \at
	\email{jhabibi@sharif.edu}
	\and 
	Hamid Haghshenas \at
	\email{haghshenas@ce.sharif.edu}
}


\date{}

\maketitle

\input{abstract}


\input{introduction}
\input{theory}
\input{Algorithms}

\input{results}
\input{conclusion}

\bibliographystyle{plainnat}      
\bibliography{bibliography}   

\end{document}

%% file: abstract.tex
\begin{abstract}
Given a random variable $O \in \mathbb{R}$ and a set of experts $E$, we describe a method for finding a subset of experts $S \subseteq E$ whose aggregated opinion best predicts the outcome of $O$. Therefore, the problem can be regarded as a team formation for performing a prediction task. We show that in case of aggregating experts' opinions by simple averaging, finding the best team (the team with the lowest total error during past $k$ rounds) can be modeled with an integer quadratic programming and we prove its NP-hardness whereas its relaxation is solvable in polynomial time.  At the end, we do an experimental comparison between different rounding and greedy heuristics on artificial datasets which are generated based on calibration and informativeness of exprets' information and show that our suggested tabu search works effectively.  \\\\
\keywords{Team Selection \and Information Aggregation \and Opinion Pooling \and Quadratic Programming \and NP-Hard}
\end{abstract}

%% file: introduction.tex
\section{Introduction}

Predicting the outcome of a random variable is an essential part of many decision making processes \citep{sprenger2012conditional}. For instance, companies have to forecast future customer demands  or changes in market regulations to do a better planning for their production \citep{attarzadeh2011software}. In some cases, lack of sufficient information (like statistical data) compels companies to seek advice from experts \citep{hammitt2013combining, chen2005information}. In order to make better informed decisions, it is logical to integrate opinions of several experts because it leads to more accurate predictions \citep{graefe2014combining}. 

In this work, we consider a situation in which a set of experts are available, each with certain level of expertise. The goal is to predict the outcome of a continuous variable O using their opinions. For each prediction task, we gather experts' opinions and aggregate them by simple linear opinion pooling. As proven before, the arithmetic average of experts’ opinions is an efficient and robust aggregation method \citep{chen2005information, chen2006predicting}. We have prediction profile of each of these experts for $k$ previous prediction tasks. The goal is to find a subset of experts with the best performance, i.e. a subset whose aggregated opinion has the least error regarding the actual outcome of $O$.

Accordingly, our method could be applied in a situation where the amount of effort required to complete a specific task in a software project needs to be predicted for effective planning and scheduling. As the relevant statistical data (data on efforts made for completing same tasks in different projects) might not be enough for a newly established company to base their predictions on, it would be justifiable to ask employees about the effort required to do the task. Suppose that efforts for doing similar tasks in previous projects of the company have been predicted by the workers. Here, our method can be applied to find a subset of workers whose aggregated opinion yield a good estimation of the effort which is crucial for successful control of software projects \citep{jorgensen2007forecasting, malhotra2011software}.

To formalize the problem, define $E = \lbrace e_1 ,\cdots e_n \rbrace$ to be  the set of experts. The $e_i$'s prediction and the actual value of $O$ in the $t$-th round  are respectively profiled by  $y_{it}$ and $x_{d}$.  In order to compare prediction ability of different subsets such as $S$ and $S'$, we use the Sum of Squared Errors (SSE) measure over the past $k$ rounds:
\begin{equation}\label{eq:SSE}
\begin{split}
f(S) &= \sum_{t=1}^k \left(\frac{\sum_{e_i \in S} y_{it}}{|S|} - x_d\right)^2
\end{split}
\end{equation} 

In the Team Selection problem, our goal is to find a subset $S$ with minimum $f(S)$. In this paper, we first consider the relaxed version of this problem where we just want to assign weights to experts and choose them fractionally. We show that this problem can be easily converted to a simple quadratic programming and therefore is polynomially solvable. Then, we show that the integer quadratic programming representing the Team Selection problem is NP-Hard (Sec. \ref{theory}). To solve this problem, we propose an augmented algorithm of the Tabu-Search used for solving the clique problem (Sec. \ref{tabusearch}). Then we suggest some other heuristics for tackling the problem and compare their precision experimentally on different artificial datasets with that of Tabu-Search and show that the Tabu-Search can give a solution to the Team Selection problem with a negligible error. In the rest of this section, some of the related works are discussed.   
\subsection{Related works}   
Various approaches for forecasting have been studied extensively. They can all be categorized into statistical and non-statistical methods. Statistical approaches require sufficient historical data to extract value patterns, whereas non-statistical approaches are based on experts’ judgments and their aggregation \citep{chen2005information, chen2006predicting}. Methods for experts’ judgments aggregation include information markets, opinion pooling, Bayesian and behavioral approaches \citep{chen2005information, clemen2007aggregating}. For information markets, scoring and compensation rules have been introduced to induce truthful forecasts and ensure participation of experts \citep{othman2010decision, boutilier2012eliciting, chen2011information, hora2007expert, zhang2012task}. Moreover, decision rules are used to exploit aggregated judgments to make a decision \citep{boutilier2012eliciting, chen2011information}. Opinion pooling and Bayesian approaches are mathematical methods for aggregating judgments to obtain accurate probability assessment for an event \citep{clemen2007aggregating, hora2007expert, genest1986combining, dani2012empirical, jacobs1995methods, french2011aggregating, morris1974decision, michaeli2008illustration}. Bayesian approach has been widely used in aggregating probability distributions with or without taking the dependence between experts into account \citep{morris1974decision, kallen2002expert, mostaghimi1996combining, mostaghimi1997bayesian}.

Expert opinion has been widely used in many fields. For safety assessment of a nuclear sector, one should rely on opinions of experts as statistical data on catastrophic events are often rare. Much the same goes for prediction of the force level and military intentions of other countries \citep{cooke1991experts}. Therefore, one of the primary applications of expert judgment is in risk analysis such as estimation of the seismic or flood risk \citep{clemen2007aggregating, clemen1999combining, reggiani2008bayesian, cooke1991experts}. 

Selecting a subset of experts who provide us with information about the outcome of an event can be regarded as forming a team of advisors. Recently, team formation, as a more general concept has received much attention. For instance, Lappas et al, took into account the cost of communication among individuals and presented two approaches for forming a team with minimum communication cost yet capable of dealing with a defined task, based on two different communication cost functions \citep{ lappas2009finding}. As another example, Chhabra et al, proposed a greedy approximation to find an optimal matching between people and some interrelated tasks by taking into account the social network structure as an indicator of synergies between members \citep{chhabra2013team}. Kargar et al, also, suggested approximation algorithms for finding a team with minimum communication and personnel costs \citep{kargar2013affordable}.

%% file: theory.tex
\section{NP-Hardness}\label{theory}
In order to explore computational complexity of the Team Selection problem, consider the following quadratic programming:
\begin{equation} \label{intprog}
\begin{split}
\text{minimize } & g(w)  =  \sum_{t=1}^k \left(\sum_{i=1}^n w_iy_{i t} - x_t\right)^2  \\
\text{subject to } & \sum_{i=1}^n w_i = 1\\
 & \forall i, w_i \in \{0,\frac{1}{m}\}
\end{split}
\end{equation}
Here, $w = (w_1, w_2, \cdots w_n)$ is the variable vector and $m$ is the number of experts to be selected.
By solving this problem for $m=1,2,\cdots n$, one can solve the Team Selection problem.

The relaxed version of the problem \ref{intprog}, where $\forall i, 0 \leq w_i \leq 1$, can be interpreted as weight assignment to each expert to indicate how much we should weigh his opinion. Thus, we refer to this problem as the Weight Assignment problem. In this section, we first show that the Weight Assignment problem is polynomially solvable by a simple quadratic programming, while its original version (the Team Selection problem) is equivalent to an NP-Hard problem. 


Define $z_{i t} = y_{i t} - x_t$ for all $1 \leq i \leq n$ and $1 \leq t \leq k$.
$z_{i t}$ is the error of the $i$-th expert's forecast in the $t$-th round.
So $y_{i t} = z_{i t} + x_t$, we have
\begin{equation}\label{eq_replace_y_with_z}
\begin{split}
g(w) &= \sum_{t=1}^k \left((\sum_{i=1}^n w_i(z_{i t} + x_t)) - x_t\right)^2 \\
&= \sum_{t=1}^k \left((\sum_{i=1}^n w_i z_{i t}) + (\sum_{i=1}^n w_i x_t) - x_t\right)^2 \\
&= \sum_{t=1}^k \left(\sum_{i=1}^n w_i z_{i t}\right)^2.
\end{split}
\end{equation}

The term inside the summation can be expanded as
\begin{equation}\label{eq_expand_inner}
\left(\sum_{i=1}^n w_i z_{i t}\right)^2 = \sum_{i=1}^n\sum_{j=1}^n w_i z_{i t} z_{j t} w_j.
\end{equation}

Replacing (\ref{eq_expand_inner}) in (\ref{eq_replace_y_with_z}) we get
\begin{equation}\label{eq_convert_to_QP}
\begin{split}
g(w) &= \sum_{t=1}^k \sum_{i=1}^n \sum_{j=1}^n w_i z_{i t} z_{j t} w_j \\
&= \sum_{i=1}^n \sum_{j=1}^n \sum_{t=1}^k w_i z_{i t} z_{j t} w_j \\
&= \sum_{i=1}^n \sum_{j=1}^n w_i (\sum_{t=1}^k z_{i t} z_{j t}) w_j \\
&= \frac{1}{2} \sum_{i=1}^n \sum_{j=1}^n w_i (2\sum_{t=1}^k z_{i t} z_{j t}) w_j.
\end{split}
\end{equation}

So the weight assignment problem can be stated as a quadratic programming
\begin{equation} \label{qcqp}
\begin{split}
\text{minimize } & \frac{1}{2} w^T Q w \\
\text{subject to } & \vec{1}^Tw = 1,\\
 & w \geq 0
\end{split}
\end{equation}
where $\vec{1}$ is the all-one vector and $Q$ is defined as
$$q_{i j} = 2\sum_{t=1}^k z_{i t} z_{j t}.$$
Clearly, $Q$ is symmetric and hence the above quadratic programming is valid. We should show that $Q$ is positive-semidefinite i.e.  for every non-zero vector $u$ we have $u^T Q u \geq 0$. Assume that $\sum_{i=1}^n u_i = c$. Define $u' = \frac{1}{c}u$. We have $\sum_{i=1}^n u'_i = 1$, thus with respect to the definition of $Q$ in (\ref{eq_convert_to_QP}) and (\ref{qcqp}), we have $g(u') = \frac{1}{2}u'^T Q u'$. So $u^T Q u = (cu')^T Q (cu') = c^2 u'^T Q u' = 2c^2 g(u')$ which is clearly non-negative (because $g(.)$ is a quadratic error function). 



We know that a quadratic programming with positive-semidefinite matrix can be solved in polynomial time 
 and hence the weight assignment problem is polynomially solvable.

 
The main result of this section is to show the NP-Hardness of the Team Selection problem. 
\begin{theorem}
The Team Selection problem is NP-Hard. 
\end{theorem}
\begin{proof}
First consider the proposed QP  (\ref{qcqp}) for the Weight Assignment problem. Adding constraints $\forall i, w_i \in \{0,\frac{1}{m}\}$ to this QP will lead to the following mathematical programming which is equivalent to the Team Selection problem (when it is solved for $m=1,2,\cdots,n$).
\begin{equation} \label{teamprog}
\begin{split}
\text{minimize } & \frac{1}{2} w^T Q w \\
\text{subject to } & \vec{1}^Tw = 1 \\
& \forall i, w_i \in \{0,\frac{1}{m}\}
\end{split}
\end{equation}
where $$q_{i j} = 2\sum_{t=1}^k z_{i t} z_{j t}.$$ 

Non-zero weight assigned to an expert means he is a member of the resulting solution. We show that this mathematical programming cannot be solved in polynomial time, unless $P = NP$. In order to prove its NP-hardness, we shall reduce the maximum independent set problem in $d$-regular graphs to this problem. Given a graph $G$, assume that $V(G) = \{v_1,v_2,...,v_n\}$ is the set of vertices of $G$, $E(G)$ is the set of its edges and $deg_G(v_i)$ denotes the $v_i$'s degree in $G$. In the maximum independent set problem, the goal is to find an empty subgraph with maximum number of vertices. We will show that every instance of the independent set problem can be transformed to an instance of the following mathematical problem which can then be reduced to the Team Selection problem:   
\begin{equation} \label{indprog}
\begin{split}
\text{minimize }& \frac{1}{2} x^T A' x \\
\text{subject to }& \vec{1}^Tx = m \\
& \forall i, x_i \in \{0,1\}
\end{split}
\end{equation}
where $A' = A + D$, $A$ is the adjacency matrix of $G$ and $D$ is a diagonal matrix with $D_{i,i} = deg_G(v_i)$. 

After solving the mathematical programming (\ref{indprog}), all the vertices with $x_i = 1$ make a subgraph $S$.
Let $i(S)$ for $S \subseteq V(G)$ denotes the number of $G$'s edges which reside in $S$. That is to say,
$$i(S) = |\{e=(x,y) \in E(G)| x, y \in S\}|.$$
First notice that 
$$x^T A' x = \sum_{i=1}^n{\sum_{j=1}^n{x_ix_jA'_{ij}}} = \sum_{i=1}^n{\sum_{j=1}^n{x_ix_jA_{ij}}} + \sum_{i=1}^n{x_i^2D_{ii}}.$$
It is easy to show that  
$$\sum_{i=1}^n{\sum_{j=1}^n{x_ix_jA_{ij}}} = 2i(S),$$
and 
$$\sum_{i=1}^n{x_i^2D_{ii}} = dm$$
Thus 
$$x^TA'x = 2i(S) + dm.$$

Minimizing $x^TA'x$ with constraint $\sum_{i=1}^n{x_i} = m$ leads to a $m$-vertex subgraph with minimum number of edges. To reduce the maximum independent set problem to the mathematical program (\ref{indprog}), it is sufficient to solve (\ref{indprog}) for all $ 1 \leq m \leq n$ and report the maximum m for which the solution is equal to dm.Thus, the mathematical programming (\ref{indprog}) is NP-Hard. 

Finally, we reduce the problem (\ref{indprog}) to the mathematical programming (\ref{teamprog}). 
It is enough to choose $z_i$s in such a way that $Q = A'$. Recall that 
$$q_{i j} = 2\sum_{d=1}^k z_{i t} z_{j t} = 2Z_i.Z_j$$ 
where $Z_i$ is a $k$-element vector composed of $z_{i d}$s. For equality of matrices, we need $q_{i j} = A'_{i j}$. In other words, we should have
\begin{equation}\label{z_constraint}
Z_i.Z_j = \begin{cases}
\frac{deg_G(v_i)}{2}& \text{if } i = j,\\
\frac{1}{2}& \text{if } v_iv_j \in E(G),\\
0& \text{otherwise}
\end{cases}
\end{equation}

To do this, first set $k = |E(G)|$ (thus $Z_i$ has a coordinate for each edge of $G$). We set $Z_{i}$'s $l$'th coordinate to $\frac{1}{\sqrt{2}}$ if $v_i$ is connected to the $l$'th edge and otherwise we set it to $0$.
To check that this assignment satisfies (\ref{z_constraint}), one can see that when $i=j$, exactly $deg_G(v_i)$ coordinates of $Z_i$ are equal to $\frac{1}{\sqrt{2}}$ and others are zero. So we have
$$Z_i.Z_i = deg_G(v_i)\times\frac{1}{\sqrt{2}}\times\frac{1}{\sqrt{2}} = \frac{\deg_G(i)}{2}.$$
When $v_i$ and $v_j$ are endpoints of an edge (say, the $l$-th edge), the $l$-th coordinate of both $Z_i$ and $Z_j$ equals to $\frac{1}{\sqrt{2}}$ and they have no other common non-zero coordinate. So we have
$$Z_i.Z_j = \frac{1}{\sqrt{2}}\times\frac{1}{\sqrt{2}} = \frac{1}{2}.$$
Finally, when $v_i$ and $v_j$ are not connected, $Z_i$ and $Z_j$ have no common non-zero coordinate and clearly
$$Z_i.Z_j = 0.$$
\end{proof}
  

%% file: Algorithms.tex
\section{Tabu Search} \label{tabusearch}

In the previous section, we showed that the Team Selection problem is NP-Hard while its relaxed version, the Weight Assignment problem, is solvable in polynomial time. In this section, we propose a tabu search algorithm to solve the Team Selection problem. 

Tabu Search has proved high performance in finding sets with specific characteristics. Different variations of this method have been used for approximating the best solution for similar problems like the Maximum Clique, Maximum Independent Set, Graph Coloring and  Minimum Vertex Cover \citep{wu2013adaptive, wu2012coloring, wu2012effective}. 
We choose the algorithm introduced in \citep{wu2013adaptive} for solving the Maximum Clique problem as a basis and transform it to an algorithm for the Team Selection problem.

Tabu Search starts from an initial solution and iteratively replaces it with one of its neighbors in order to get closer to the optimal solution. 
In each iteration, a local search is done for finding a group whose collective prediction has the least error. If there is no such neighbor, current solution is regarded as a local minimum. To escape from local minimums, Tabu Search allows the least worse neighbor to be selected. Wu \& Hao use Probabilistic Move Selection Rule (PMSR) when no improving solution is found in neighborhood. This strategy helps to move to other neighbors when the quality of the local minimum is much less than that of the optimal solution \citep{wu2013adaptive}. We use a similar strategy in our proposed algorithm. For preventing previous solutions from being revisited, Tabu Search uses a tabu list which records the duration of each element being kept from moving into or out of current solution. 
\begin{algorithm}
\caption{Tabu Search For Team Selection Problem}\label{TB}
\begin{algorithmic}[1]

\Require A Set of experts ($E$), Expert's sequence of past predictions, integer $MaxIter$ (Maximum number of successive tries which fail to find better solution), $m$ (size of the team)
\Ensure A team with minimum SSE if found

\State $\textit{S} \gets MaxWeightsAssignedTo(E,m)$ 
\State $\textit{lowerBoundOfSSE} \gets g(w)$ $\{$ w contains weights assigned to $E$ $\}$ 
\State $i \gets 0$  $\{$ number of iterations $\}$
\State $\textit{bestSet} \gets \textit{S}$ $\{$ Records the best solution found so far $\}$ 
\While {$i < MaxIter$}
\State $S' \gets S\cup \{v\} \setminus \{u\}$ with minimum SSE among all $u,v$ pairs not in tabu list
\If{$f(S') < f(S)$}
\State $S \gets S'$
\Else
\State $\textit{S} \gets S' \text{with probability }1-P$
\State $\text{ or a random neighbor with probability } P$ 
\EndIf
\State \text{Update the tabu list} $\{$ List of all $u,v$ pairs which are tried in iterations$\}$
\If {$f(S) = \textit{lowerBoundOfSSE}$}
\State \Return $S$
\EndIf
\If {$f(S) < f(bestSet)$}
\State $\textit{bestSet} \gets S$
\State $i \gets 0$
\Else
\State $i \gets  i+1$
\EndIf
\EndWhile
\State \Return \textit{bestSet}
\end{algorithmic}
\end{algorithm}    

Algorithm \ref{TB}, shows the pseudo code of our proposed tabu search. The first line shows the initialization of the first set (team), which then goes through improvements in the main loop.  As the initial set can play an important role in Tabu Search performance \citep{wu2013adaptive}, we suggest  the initial set to be equal to the set of $m$ experts who are given the largest weights in an optimum solution for the Weight Assignment problem (this is shown by $MaxWeightsAssignedTo(E,m)$). 
 
In each iteration of the loop, the amount of improvement gained by each possible swap is calculated simply by subtracting SSE of the team resulting from swapping two experts (one in the current set with another out of it) from the SSE of the current team. If the best possible swap results in a better solution (lower SSE), then the current set is updated with the new solution. Otherwise, a random set is selected as the current solution with probability $P$. In another word, $P$ is the probability of escaping from a local minimum. Like various kinds of Tabu Search, we use tabu list to prevent producing repeated sets. Therefore, after substituting a member with another expert out of the current set, tabu list is updated with regard to tabu tenure values calculated for both selected experts. This implies that for some time these experts are not allowed to move in or out of the current set in next iterations.

There are two terminating conditions for this algorithm. For one, the main loop terminates by not finding any better set after $maxIter$ successive iterations. For another, when the current solution is equal to the solution of the Weight Assignment problem, the algorithm stops the search process. That is to say, there is no other set with less SSE.

%% file: results.tex
\section{Comparision}
In this section, inspired from algorithms proposed for similar problems, we suggest different heuristics for the Team Selection problem and compare their efficiency with the tabu search proposed in Section \ref{tabusearch}.
\subsection{Heuristics}

\noindent\textbf{Random Rounding}: 
Random rounding defines a threshold ($T$) and selects experts with weights above the threshold with probability $P$ and the others with probability $1 - P$. This process will continue until $m$ experts are selected.  Our experiments show that higher amounts of $T$ yields better results.  

\noindent\textbf{Max-Weights}: 
This rounding algorithm takes the $m$ experts with largest weights as members of the team. 

\noindent\textbf{Min-Effect}: 
In each round, this algorithm tries to find a member who has the minimum effect on the SSE of $E$. According to the equation (\ref{eq_replace_y_with_z}), the effect of each person on the SSE function can be calculated as the following: 
\begin{equation} \label{expert-effect}
\left(2w_i \sum_{j \neq i} w_j \sum_{d = 1}^k z_{id}z_{jd}\right) - w_i^2\sum _{d = 1}^k z_{id}^2,
\end{equation}
which is equal to sum of the terms including $z_i$.\\
\noindent\textbf{Best Pairs}: 
Despite the fact that experts with high prediction error are not desirable, aggregated opinions of two or more of them may have an acceptable error. This is due to the bracketing concept \citep{graefe2014combining}. Thus, in this algorithm we allow pairs whose aggregated opinion has the minimum absolute error to be selected. The algorithm computes sum of the absolute errors of the aggregated opinions of all pairs over past $k$ rounds, then report $\lfloor \frac{m}{2} \rfloor$ of pairs with smallest calculated values. For odd values of $m$, last person would be the one among remained experts with minimum sum of absolute errors.

\noindent\textbf{Remove Least Weights}: 
This algorithm runs the Weight-Assignment problem's algorithm iteratively and removes one with the least weight in each iteration. The process continues until $m$ experts are remained. 

\noindent\textbf{Minimum Error}: 
One simple strategy of members selection is to find experts with minimum sum of absolute errors during past $k$ rounds. For simplicity, we call this approach "Minimum Error".

%
%
%
%

\subsection{Comparison of Algorithms}

In this section, we evaluate the tabu search and other heuristics for solving the Team Selection problem. We consider four different simulation scenarios  with 15 experts, each with known distribution for their predictions and tested the algorithms for team sizes from 2 to 10. These scenarios are based on two measures for evaluating quality of expert's distribution (calibration and informativeness introduced by Hammitt \& Zhang in \citep{hammitt2013combining}) and are described as follows:
\begin{itemize}
\item
\textbf{Normal1: }In this case, random variable $O$ and experts' beliefs have normal distribution with $\mu = 10$, thus, experts' information are calibrated. Standard deviation of each expert's distribution is randomly selected from $[1,2]$.
\item
\textbf{Normal2: }This case models calibrated but less informative experts. Therefore, like the previous case, all distributions are normal with $\mu = 10$, but this time, standard deviations of experts' predictions are between $1$ and $7$ ($\sigma_i$ is randomly selected from $[1,7]$).
\item
\textbf{Normal3: }In the third case, we simulate a situation in which some of the experts are not calibrated. For doing this, experts' beliefs have normal distribution with random means that are selected uniformly from $[8,12]$. Like Normal1, standard deviations are chosen randomly between 1 and 2. 
\item
\textbf{Exp: }For the final case, we simulate both the reality and the experts' predictions with exponential distributions with $\mu = 10$.

\end{itemize}

\begin{figure}
\begin{center}
\includegraphics[width=12.2cm , height = 7.5cm]{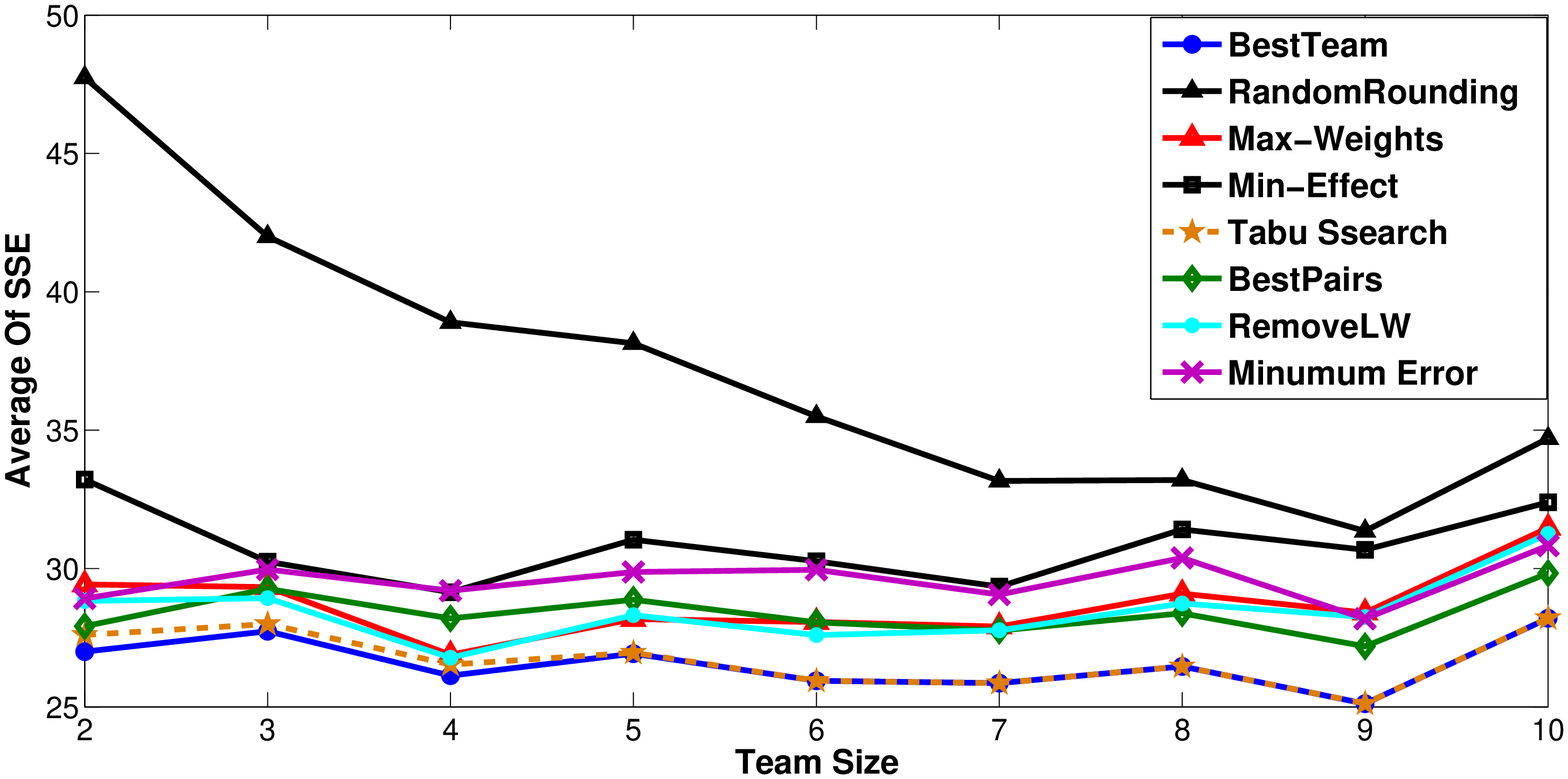}
\end{center}
\captionsetup{labelfont=bf,labelsep=space}
\caption{ Comparison of the algorithms and heuristics for the case Normal1}
\end{figure}

\begin{figure}
\begin{center}
\includegraphics[width=12.2cm , height = 7.5cm]{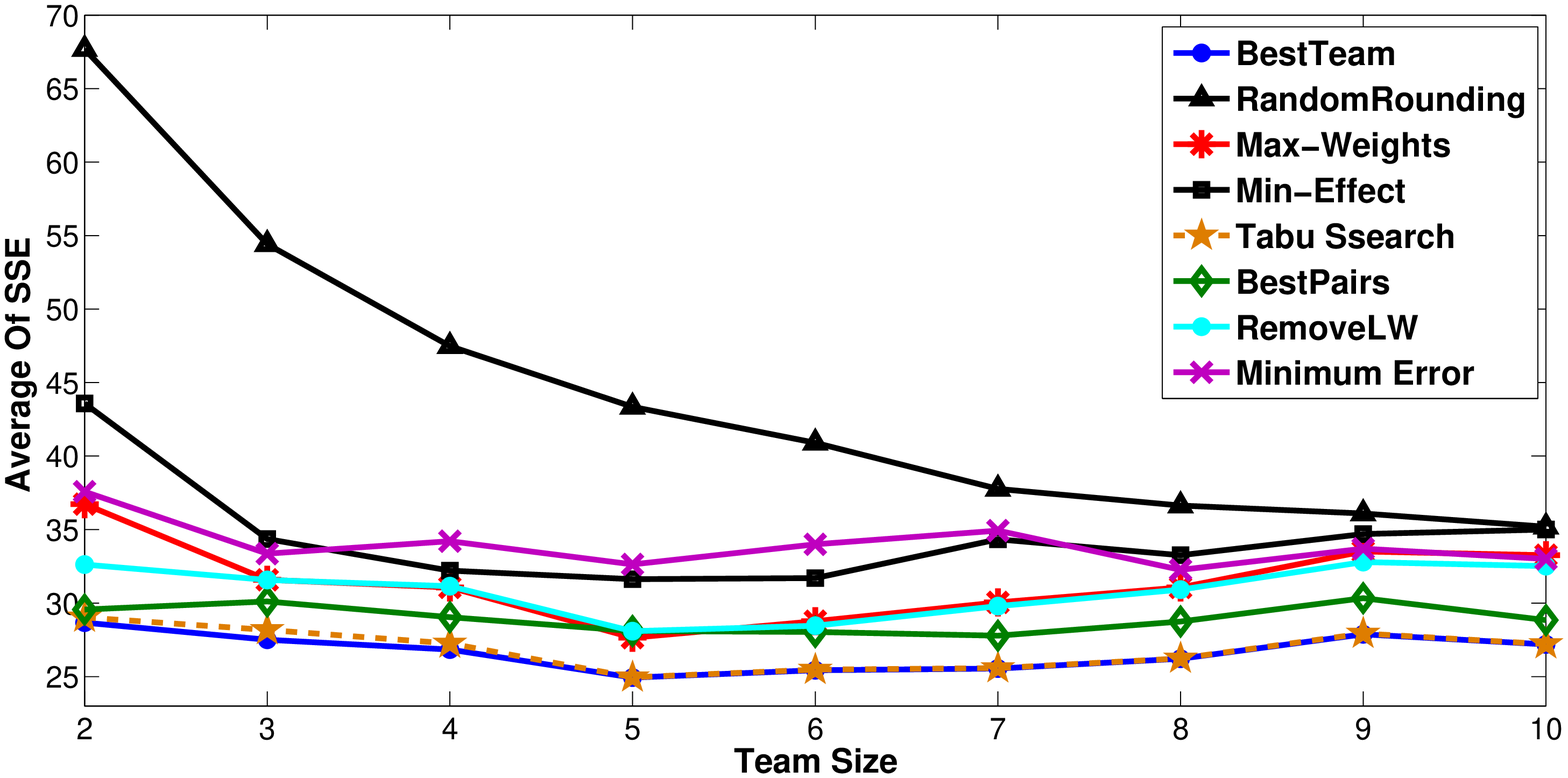}
\end{center}
\captionsetup{labelfont=bf,labelsep=space}
\caption{ Comparison of the algorithms and heuristics for the case Normal2 }
\label{fig:rc}
\end{figure}

\begin{figure}
\begin{center}
\includegraphics[width=12.2cm  , height = 7.5cm]{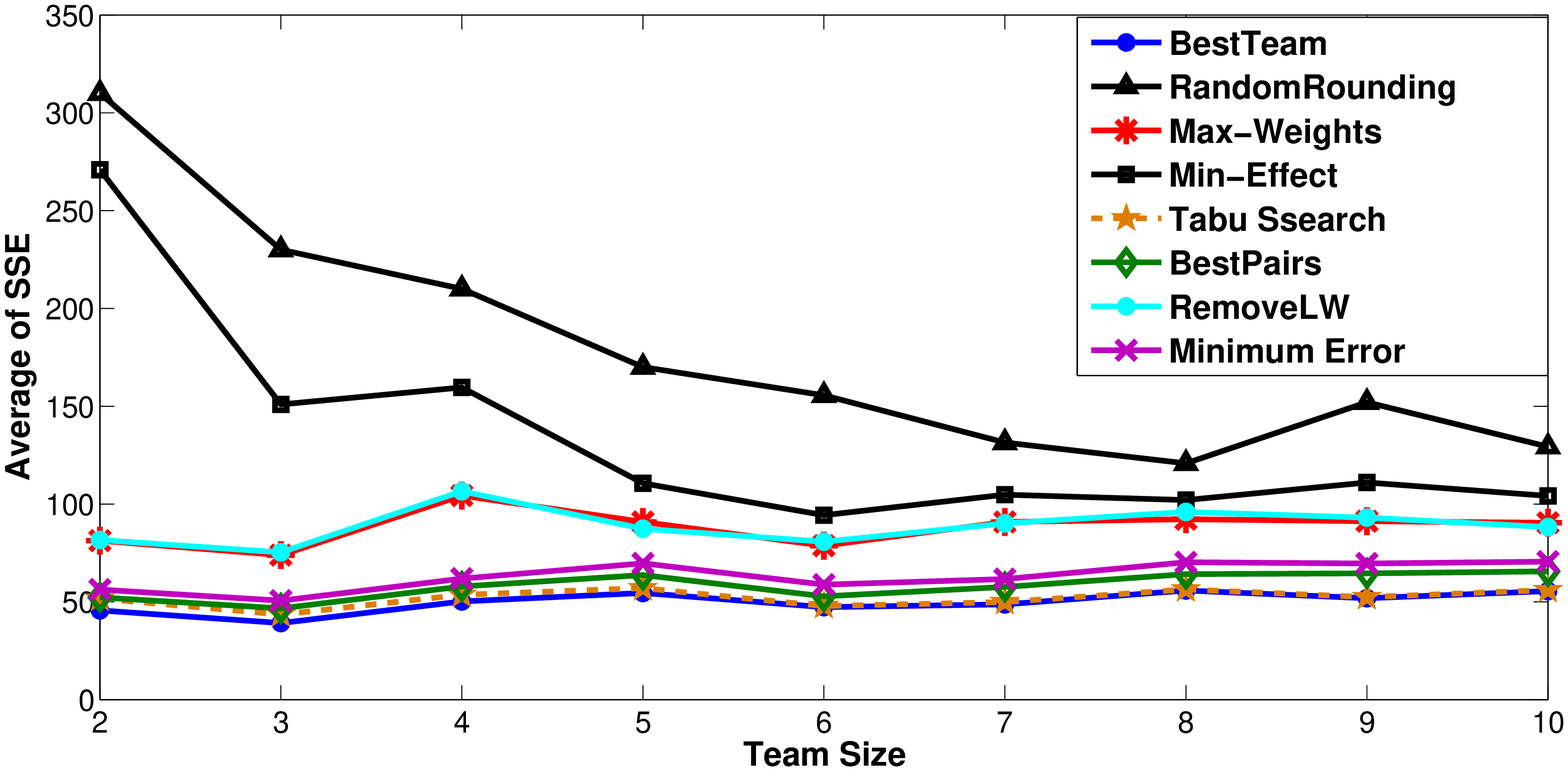}
\end{center}
\captionsetup{labelfont=bf,labelsep=space}
\centering
\caption{ Comparison of the algorithms and heuristics for the case Normal3 }
\end{figure}

\begin{figure}
\begin{center}
\includegraphics[width=12.2cm  , height = 7.5cm]{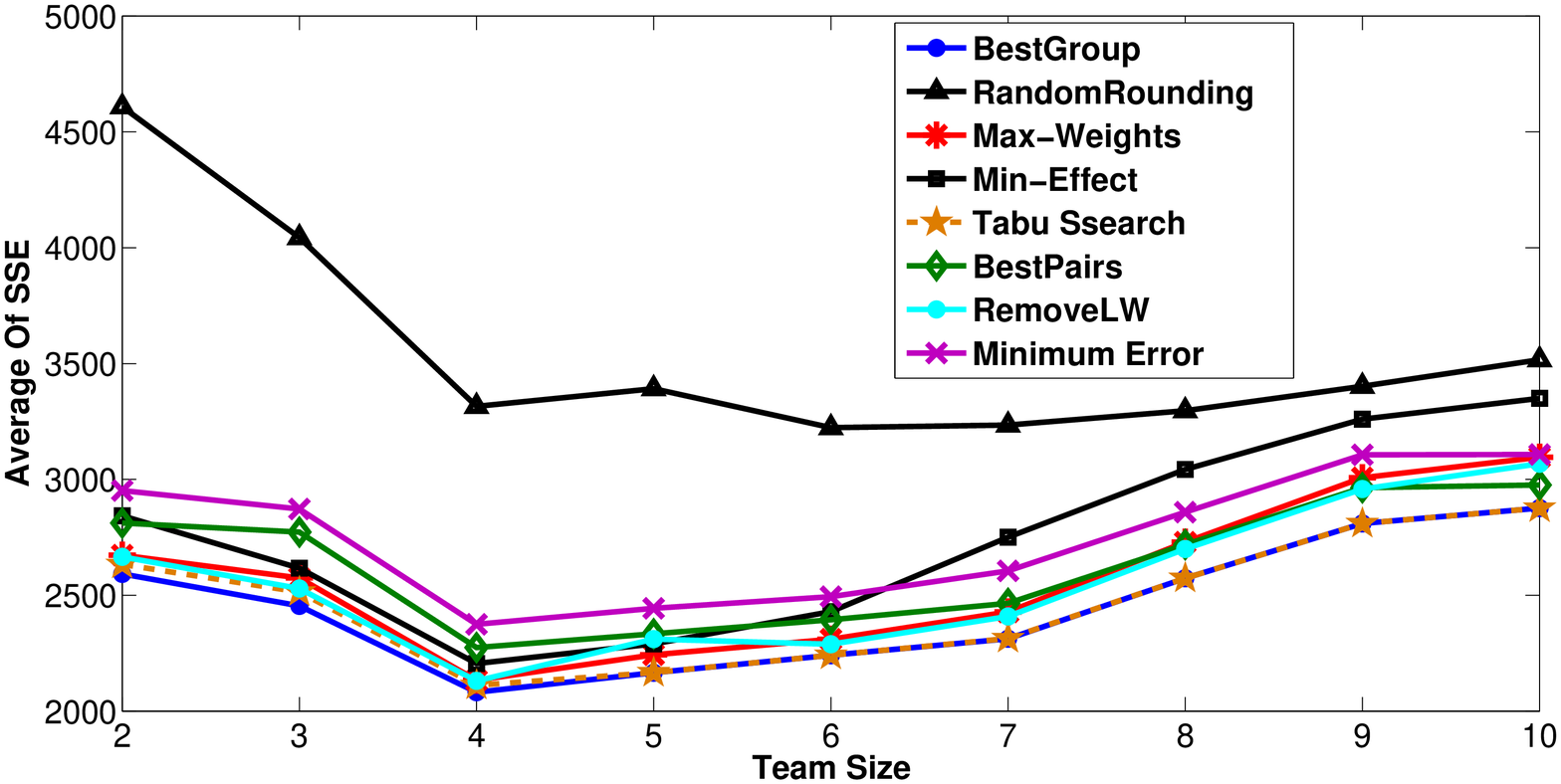}
\end{center}
\captionsetup{labelfont=bf,labelsep=space}
\centering
\caption{ Comparison of the algorithms and heuristics for the case Exp }
\label{fig:gc}
\end{figure}

As presented in Figures 3 to 6 and Table 1, Tabu Search produces far better results in all cases.  It can also be seen that the result of this algorithm is very near to the best possible algorithm which tries all the feasible solutions and return the best one. This means that our suggested algorithm is less sensitive to the distribution of the event $O$. Thus, Tabu Search is more reliable than other proposed heuristics. 
 Best Pairs's efficiency for normal distributions is comparable with Tabu search. Moreover, the average execution time of Best Pairs is around 0.02 of Tabu Search (Table \ref{comT}). Therefore, it would be an acceptable method for quickly forming a team. However, in the case of exponential distributions, Best Pairs performance for small teams is even worse than Min-Effect which is due to the increase in diversity of the forecasts. Hence, the probability of neutralization of an expert's error by another, decreases. 
\begin{center}
\begin{table}[!htb]
\captionsetup{labelfont=bf,labelsep=space}
\caption{Average of difference between the SSE of the best team and the SSE of the heuristics' solutions for teams of size 2 to 10}\label{comF} 
    \begin{tabular*}{1\linewidth}{@{\extracolsep{\fill}}|c|c|c|c|c|}
    \hline
     \textbf{Algorithm}&\textbf{Normal1}&\textbf{Normal2}&\textbf{Normal3}&\textbf{Exponential}\\
     \hline
	RandomRounding&10.348&128.937&17.702&1102.74\\
     \hline
	Max-Weights&2.153&38.329&4.831&121.021\\
     \hline
	Min-Effect&4.26&84.39&7.84&297.97\\
     \hline
	Tabu Search&0.145&2.186&0.18&14.833\\
     \hline	
	BestPairs&1.788&8.507&2.26&178.273\\
     \hline	
	RemoveLW&1.897&38.895&4.188&105.951\\
     \hline
\end{tabular*}
\end{table}
\end{center}
\begin{center}
\begin{table}[!htb]
\begin{center}
\captionsetup{labelfont=bf,labelsep=space}
\caption{Comparison of the heuristics' average of the execution time for finding best team of 8 experts amog 15 experts in case of Normal2}\label{comT} 
\begin{tabular*}{0.55\linewidth}{@{\extracolsep{\fill}}|c|c|}
 \hline   
 \textbf{Name of Algorithm}&\textbf{Average Exe.Time}\\
    \hline
     BestTeam&27.3742\\
     \hline
     RemoveLW&0.0461\\
     \hline
     Tabu Ssearch&0.0345\\
     \hline
      Min-Effect&0.0036\\
     \hline
     Max-Weights&0.0035\\
     \hline
     RandomRounding&0.0013\\
     \hline
     BestPairs&0.0009\\
     \hline
\end{tabular*}
\end{center}
\end{table}
\end{center}
\subsection{Other Experiments}
\noindent\textbf{The effect of the team size}:  As the number of hired experts determines the cost incurred, we would like to know the effect of the team size on the accuracy of the aggregated opinion of its members. Therefore, in our simulations we capture the accuracy for different team sizes and depict the results in Figure \ref{fig:size}. This figure shows the optimal solution for different sizes of $E$. The results show that increasing the number of experts first reduces but then increases SSE again. Therefor, we can conclude that large values for $m$ is neither cost effective nor efficient. 
\begin{figure}[!htb]
\centering
\includegraphics[width=12.2cm  , height =6cm]{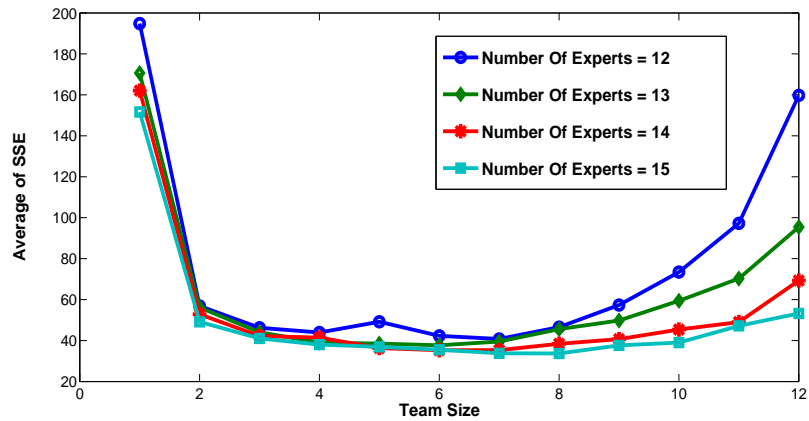}
\captionsetup{labelfont=bf,labelsep=space}
\caption{Eeffect of the team size on the average of the SSE of the best teams when 12, 13, 14 or 15 experts are available}
\label{fig:size}
\end{figure}

\begin{figure}[!htb]
\includegraphics[width=12.2cm  , height = 6cm]{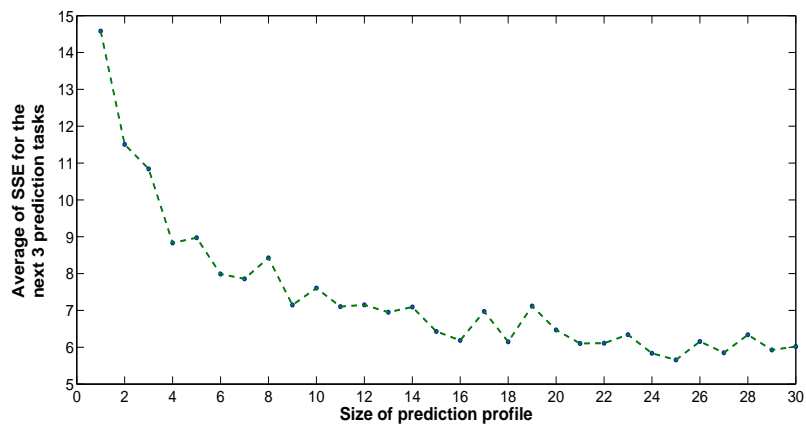}
\captionsetup{labelfont=bf,labelsep=space}
\caption{Effect of the prediction profile size on the SSE of the best team for the next 3 prediction tasks}
\label{fig:past}
\end{figure}
\noindent\textbf{The effect of the prediction profile size}: It is apparent that having more information about the experts' past predictions, improve the quality of the final result. The question here is how much would be enough. We observe that for large values of the size of the experts' prediction profile, the decrease in SSE will finally stop. Therefore, the first point with minimum value would be the optimal number of past records. The results of this experiment can be seen in Figure \ref{fig:past}.  

%% file: conclusion.tex
\section{Conclusion}
In this paper, we addressed the Team Selection problem in which we wanted to form a team of experts with minimum error for performing a prediction task. To simplify the problem, we first studied the relaxed version of the problem (the Weight Assignment problem) in which our goal was to find the best weights for linear opinion pooling. We proved that this problem can be solved with a simple quadratic programming in polynomial time. Then we proved that the Team Selection problem is NP-hard. In the rest of the paper, we proposed a tabu search algorithm for solving the problem. Our experiments show the superior accuracy of this algorithm compared to other proposed algorithms. It is also shown that the accuracy of this algorithm is comparable to the best possible algorithm. 